\documentclass[11pt,transaction,twocolumn]{IEEEtran}
\usepackage{graphics}
\usepackage{graphicx}
\usepackage{amsmath,amsfonts,amssymb,varioref,mathrsfs,verbatim}
\usepackage{latexsym}
\usepackage{epsfig}
\usepackage{epstopdf}
\usepackage[nospace,noadjust]{cite}
\usepackage{array}
\usepackage[english]{babel}
\usepackage{color}
\usepackage{babel}

\newtheorem{lemma}{Lemma}

\newcommand{\BlackBox}{\rule{2.6mm}{2.6mm}}
\newenvironment{proof}{\noindent {\bf Proof:}~}{\hspace*{\fill}\(\BlackBox\)}

\include{new_commands}

\makeatletter
\let\l@ENGLISH\l@english
\makeatother

\begin{document}

\title{On the Feedback Reduction of Relay Aided Multiuser Networks using Compressive Sensing}
\author{Khalil M. Elkhalil, Mohammed E. Eltayeb,~\IEEEmembership{Student Member,~IEEE,} Abla Kammoun,~\IEEEmembership{Member,~IEEE}, Tareq~Y.~Al-Naffouri ,~\IEEEmembership{Member,~IEEE}, and Hamid~Reza~Bahrami,~\IEEEmembership{Member,~IEEE}%

\thanks{Khalil Elkhalil, Abla Kammoun and Tareq Y. Al-Naffouri are with the Electrical Engineering Department, King Abdullah University of Science and Technology, Thuwal, Saudi Arabia; e-mails: \{khalil.elkhalil,abla.kammoun, tareq.alnaffouri\}@kaust.edu.sa. Tareq Y. Al-Naffouri is also associated with the Department of Electrical Engineering, King Fahd University of Petroleum
and Minerals, Dhahran 31261, Kingdom of Saudi Arabia. 

Mohammed Eltayeb is with the department of Electrical and Computer Engineering, The University of Texas at Austin, Texas, USA ; e-mail: meltayeb@utexas.edu. 

Hamid Reza Bahrami is with the Department of Electrical and Computer Engineering, The University of Akron, Ohio, USA; e-mail: hrb@uakron.edu.
}


}

\maketitle

\vspace{-15mm}
\begin{abstract}
In this paper, we propose a feedback reduction scheme for full-duplex relay-aided multiuser networks. The proposed scheme permits the base station (BS) to obtain channel state information (CSI) from a subset of strong users under substantially reduced feedback overhead. More specifically, we cast the problem of user identification and CSI estimation as a block sparse signal recovery problem in compressive sensing (CS). Using existing CS block recovery algorithms, we first obtain the identity of the strong users and then estimate their CSI using the best linear unbiased estimator (BLUE). To minimize the effect of noise on the estimated CSI, we introduce a back-off strategy that optimally backs-off on the noisy estimated CSI and derive the error covariance matrix of the post-detection noise. In addition to this, we provide exact closed form expressions for the average maximum equivalent SNR at the destination user. Numerical results show that the proposed algorithm drastically reduces the feedback air-time and achieves a rate close to that obtained by scheduling schemes that require dedicated error-free feedback from all the network users.
\end{abstract}

\section{Introduction}
 In a relay aided networks, a relay station (usually low powered) is deployed to assist communication between the base station (BS) and a destination user. Relaying techniques can be classified, based on their forwarding strategy and required processing at the relay terminal, as \emph{decode and forward} (DF) or \emph{amplify and forward} (AF) \cite{rl2}. Relay transmission for both AF and DF relaying can be performed in \emph{half-duplex (HD)} (e.g. \cite{relay11a}) or \emph{full-duplex (FD)} (e.g. \cite{FDAF} ) mode. In HD mode, the BS and the relay transmit on orthogonal channels, whereas in FD mode,  the BS and relay share a common channel and the relay transmits and receives simultaneously over the same channel. Thus, half-duplex relay schemes induce a throughput loss due to the pre-log factor 1/2 \cite{six}. Full duplex relay schemes are shown to improve the network throughout by eliminating this pre-log factor \cite{FDAF}. 

In a network with large number of users, the BS can exploit the multiuser diversity to maximize the network capacity. This is achieved by requesting all users to feed back their equivalent SNR (BS-relay-user) to the base station (BS). Clearly, this is expensive in terms of spectrum utilization and results in a large feedback overhead, especially when the number of users is high \cite{jindal},\cite{ammimo}. However, channel state information (CSI) mismatch that could arise from both channel estimation errors and quantized feedback results in a loss in the achievable rate. Therefore, for practical scenarios, feedback strategies should be able to reduce the overhead while satisfying certain performance guarantees.
In fact, to reduce the feedback, it is possible to request only a few strong users to encode their SNR on non-orthogonal codewords. This creates a sparse user regime, where sparse approximation algorithms (see e.g. \cite{cP08}) are then applied to recover the identity and SNR of the strong users. The work in \cite{cs12} uses Compressed Sensing (CS) to reduce the feedback load of single-input-single-output (SISO) networks, while other work e.g. \cite{me}-\cite{cs15} apply CS to reduce the feedback overhead for MIMO networks. \\
In this paper, we apply CS to reduce the feedback overhead of a relay-aided multiuser network. We exploit the feedback nature of FD relays to pose the feedback problem as a block CS recovery problem. Block CS recovery leverages the a priori information of the signal block size to better differentiate true signal information from recovery artifacts. This leads to a more robust recovery \cite{baraniuk}. To the best of our knowledge, this is the first paper that proposes a CS-based feedback algorithm for multi-user FD relay-aided networks. We focus on networks that are comprised of a BS serving users in their respective coverage areas. It is assumed that there is no direct link between the BS and the users, and hence, communication must take place via a relay terminal. Examples of such a topology include cell-edge users, users that are shadowed from the BS, e.g. users in a shopping complex or an air-port, or users covered by a pico or a femto cell as the case in heterogeneous networks \cite{het1}\cite{het2}. To exploit the multi-user diversity, the BS requires feedback of CSI from all the users. Based on the feedback the BS receives (via the relay), it schedules the user with the best BS-relay-user channel condition. This feedback requirement poses two immediate challenges. Firstly, CSI feedback (users-relay-BS) generates a great deal of feedback overhead (air-time) that could result in significant performance hits if the number of users becomes too high. Secondly, the fed back channel information is usually corrupted by additive noise. In this paper we addresse the above challenges and provide the following contributions

\begin{itemize}
\item We propose a CS-based full-duplex feedback algorithm that reduces the feedback overhead of multi-user relay aided networks.
To the best of our knowledge, this is the first paper that introduces full-duplex feedback and accounts for the relay self-interference on the feedback link.
\item We account for the feedback noise and refine the estimated equivalent SNR, and optimally \emph{backing off}\footnote{The term back-off denotes an SNR reduction, applied by the BS to the estimated SNRs, to reduce the likelihood of an SNR over-estimation error.} on the noisy equivalent SNR estimates.
\item We derive the post-CS detection noise variance.
\item We derive the average equivalent SNR of the best user when the relay employs HD and FD feedback.
\end{itemize}

The remainder of the paper is organized as follows. In Section \ref{sysmod}, we introduce the system model. We describe the proposed feedback algorithm in Section \ref{proposed} and evaluate its performance in Section \ref{PA}. In Section \ref{nr}, we present some numerical results prior to concluding our work in Section \ref{conclusion}. 

{\bf{Notations:}} Throughout this paper, we use the following notations : $F_X\left(.\right),$ $F^{-1}_X\left(.\right)$, $p_X\left(.\right)$ and $\mathbb{E}\left(X\right)$ stand for the cumulative density function (CDF), the inverse CDF, the probability density function (PDF), and the expectation of the random variable $X$ respectively. We denote by $\mathbb{P}\left(\mathcal{A}\right)$, the probability of the event $\mathcal{A}$.  Matrices are denoted by bold capital letters, rows and columns of the matrices are referred with lower case bold letters. We use the subscript $c$ and $r$ to refer to a matrix column and row respectively, for example $\mathbf{a}_{c,i}$ and $\mathbf{a}_{r,i}$ are respectively the $i$th column and row of the matrix $\mathbf{A}$. We denote by $\mathbf{A}_{\mathcal{I}}$, the submatrix of $\mathbf{A}$ whose columns are in $\mathcal{I}$. $tr\left(.\right)$ and $\left[.\right]^t$ respectively denote the trace and the transpose operators on matrices. To refer to a matrix $\mathbf{A}$ as a positive semi-definite matrix, we use the notation $\mathbf{A}\succeq 0$, consequently $\mathbf{A} \succeq \mathbf{B}$ if and only if $\mathbf{A}-\mathbf{B}\succeq 0$. We also denote by $\Gamma\left(.\right)$, $\gamma\left(.,.\right)$ and $\Gamma\left(.,.\right)$, the Gamma function, the lower incomplete Gamma function and the upper incomplete Gamma function respectively.

\section{System Model}\label{sysmod}
In this section, we introduce the downlink and feedback (uplink) transmission models for the relay-aided multi-user network. 

\begin{figure}[t!]
        \centering
    \includegraphics[width=3.5in]{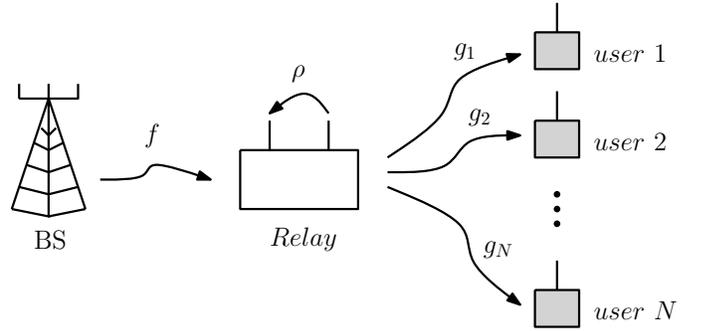}
  \caption{Network Model. }
    \label{fig:model}
  \end{figure}
\subsection{Downlink Transmission Model}
We consider a network with $N$ cell-edge users served by one BS via a single FD relay (no direct path exits between the BS and destination users). As shown in Fig. \ref{fig:model}, both the BS and the users are equipped with a single antenna, whereas the relay is equipped with one transmission antenna and one receive antenna. The two relay antennas operate simultaneously over the same frequency, i.e. reception and transmission occur at the same time over the same frequency. We assume that a direct link exits between the BS and the relay, hence, the channel from the BS to the relay $f$ is modeled by a Rician distribution with fading parameter $K=\frac{b^2}{2\sigma^2}$, where $2\sigma^2$ is the average power in the Non-Line Of Sight (NLOS) multipath components and $b^2$ is the power in the LOS component \cite{andrea}. For simplicity, the BS-relay channel can be approximated to follow a Nakagami fading channel parameterized by  $P=b^2+2\sigma^2$ and fading parameter $d=\frac{\left(K+1\right)^2}{2K+1}$. Therefore,
\begin{equation}
F_{\left|f\right|^2}\left(x\right)=\frac{\gamma\left(d,\frac{ x}{\theta}\right)}{\Gamma\left(d\right)},
\end{equation}
where $\theta=\frac{P}{d}$. The channels from the relay to the $N$ users, $\left ( g_n \right )_{n=1,2,...,N}$ are assumed to be independent and identically distributed (i.i.d.) zero-mean complex Gaussian random variables with variance $\sigma^2_g$. To maximize the downlink link rate, the destination user is selected based on the following rule
\begin{equation}\label{SNR}
n^* = \arg \max_n \gamma_n,
\end{equation}
where $\gamma_n = \frac{P_r|g_n|^2}{N_0}$ is the instantaneous receive SNR at the $n$th user which is an exponential random variable with mean $\bar{\gamma}=\frac{P_r \sigma_g^2}{N_0}$, and $N_0$ is the additive noise variance. 
 The relay may apply cancellation techniques \cite{e1},\cite{e5} to (partly) eliminate this self-interference. Usually this involves estimation of the loop channel $h$ and then subtraction of the interfering signal. However, in practice, complete interference cancellation is not achieved due to non-ideal channel estimation and signal processing. In this paper, we model the effective gain of the loop channel due to the residual interference after cancellation as a multiplicative factor $\rho$ that remains constant during the transmission interval.\footnote{This assumption can be justified by the fact that relay's antennas are fixed during transmission.}  Let $s_{n^*}(t)$ and $\hat{s}_{n^*}(t)$ represent the original and the decoded symbol by the relay respectively, then after selecting the destination user $n^*$, the BS broadcasts the symbol $s_{n^*}(t)$ which is received by the relay as
\begin{equation}\label{rr}
y_r (t) = \sqrt{P_s}f s_{n^*}(t) + \rho \sqrt{P_r}\hat{s}_{n^*}(t-1) + z_r(t),
\end{equation}
where $P_s$ and $P_r$ represent the BS and the relay transmission powers respectively, and the term $z_r$ represent the zero mean variance $N_0$ additive Gaussian noise at the relay. During the second hop, the relay will decode and forward\footnote{Throughout the paper, we only consider decode and forward as the downlink protocol.} the symbol, and the ${n^*}$th user receives
\begin{equation}\label{ru}
y_{n^*}(t) = \sqrt{P_r}g_{n^*} \hat{s}_{n^*}(t) + z_{n^*}(t),
\end{equation}
where $z_{n^*}(t)$ is the zero mean variance $N_0$ additive Gaussian noise at the $n^*$th user.
 The transmitted symbols are assumed to have the same constant normalized transmitted power over time, i.e. $\mathbb{E}[|s(t)|^2]=1$.\\

\begin{figure}[t!]
        \centering
    \includegraphics[width=3.5in]{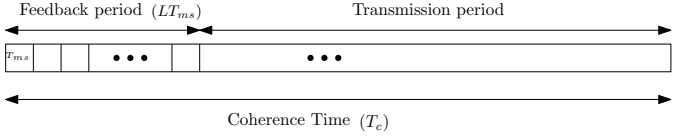}
\label{model}
\caption{Structure of a frame of duration $T_c$ composed of feedback period and transmission period.}
\end{figure}

\subsection{Feedback Model}
We assume a reciprocal TDD setup so that the BS-relay and the relay-user channels on the uplink are the same as those of the downlink. As shown in Figure 2, the feedback channel is assumed to be slotted, shared, and opportunistically accessed by the users. There are $L$ feedback mini-slots and each feedback transmission is received over a duration of one feedback mini-slot time $T_{ms}$ (see Figure 2). All channel coefficients, i.e. $f$, $h$ and $\left(g_n\right)_{n=1,2,...,N}$, are assumed to be constant during the feedback and data transmission periods, (i.e. all channels are fixed during one coherence interval of duration $T_c$). Each user is assigned an ID which is a unique signature sequence vector of dimension $M \ll N$ for use on the slotted feedback channel. The choice of $M$ is discussed in Section \ref{cs1}.  The signature sequence vectors, are drawn from the columns of a real Gaussian matrix $\mathbf{A} \in \mathbb{R}^{M\times  N}$ (with zero mean and variance $\frac{1}{M}$ i.i.d. entries), and are deterministically assigned by the BS to the users.

Prior to feedback, the relay broadcasts a pilot from which the BS and the users estimate their channels with the relay. Only users having an instantaneous receive SNR (Please refer to eq. (\ref{SNR})) higher than a predetermined threshold $\gamma_{th}$ encode (or multiply) their SNRs (relay-user SNR) with their signature sequence vectors and transmit the combination to the BS (via the relay) after applying proper uplink channel compensation, i.e. the $n$th user scales its transmission by $\frac{1}{g_n}$. The remaining users remain silent or effectively feed back a ``0''.  Let $\mathbf{x}=\left[x_1 \quad x_2 \quad ... \quad x_N\right]^t$ denotes the fed back sparse vector, where $x_n$ is the fed back value of the $n$th user and it is given by
\begin{equation}
x_n=\left\{\begin{matrix}
\gamma_n&, \gamma_n > \gamma_{th}\\
 0&, \text{otherwise.}
\end{matrix}\right.
\end{equation}
 Also, let $\mathbf{a}_{r,m}=\left[a_{m,1} \quad a_{m,2} \quad ... \quad a_{m,N}\right] \in \mathbb{R}^{1\times N}$, be the $m$th row of $\mathbf{A}$. Then at the $m$th feedback mini-slot, the $n$th user scales its feedback information by $a_{m,n}$ and feeds back the combination to the relay which simply forwards what it receives to the BS.  The relay can feed back the received measurements using FD or HD modes. In the HD case, the relay receives the following at the $m$th feedback mini-slot
\begin{eqnarray} \label{hdreceive}
y_r^{(m)}&=&\mathbf{a}_{r,m}\mathbf{x}+z_r^{(m)},
\end{eqnarray}
In the case of FD feedback, the BS will have to account for the relay self-interference which does not exist in the case of HD feedback. 
\begin{eqnarray}\label{yrm}
y_r^{(1)}&=&\mathbf{a}_{r,1}\mathbf{x}+z_r^{(1)}\\
y_r^{(m)}&=&\mathbf{a}_{r,m}\mathbf{x}+\rho y_r^{(m-1)} +z_r^{(m)}, m=2,3,...,M
\end{eqnarray}
where $y_r^{(m)}$ is the received signal at the relay at the $m$th feedback mini-slot, $\rho y_r^{(m-1)}$ is the self-interference residual at time $m$, and $z_r^{(m)}$ is the relay noise component at the $m$th mini-slot. For the sake of simplicity, we assume that $\rho$ is small enough such that $\left|\rho\right|^k \approx 0, \forall k\geq J$, where $J$ is a constant, so that all factors multiplied by $\left|\rho\right|^k$ will be ignored for $k \geq J$. With this in mind, we can transform the recursive expression in (\ref{yrm}) and (7) to the following form
\begin{equation}\label{yrm2}
\begin{split}
y_r^{(m)}&=\sum_{k=0}^{\min\left(J-1,m-1\right)}\rho^k\mathbf{a}_{r,m-k}\mathbf{x} \\
&+\sum_{k=0}^{\min\left(\lceil\frac{J}{2}\rceil-1,m-1\right)} \rho^k z_r^{(m-k)}.
\end{split}
\end{equation}
The term $\frac{J}{2}$ in the second term of (\ref{yrm2}) is due to the assumption of the convergence of the noise variance to zero for values of $k \geq \lceil\frac{J}{2}\rceil$, i.e. $\mathbb{E}[|\rho^{\frac{J}{2}} z_r|^2] \approx 0$. After receiving the $m$th feedback measurement, the relay forwards that signal to the BS. After equalizing by the BS channel, the received $m$th measurement at the BS becomes
\begin{equation}
\label{ys}
\begin{split}
y_s^{(m)}&=\sum_{k=0}^{\min\left(J-1,m-1\right)}\rho^k\mathbf{a}_{r,m-k}\mathbf{x} \\
&+\sum_{k=0}^{\min\left(\lceil\frac{J}{2}\rceil-1,m-1\right)} \rho^k z_r^{(m-k)}+\frac{w^{(m)}}{f},
\end{split}
\end{equation}
where $w^{(m)}$ is the noise term at the BS at the $m$th reception instant.

%


%
\section{Opportunistic User Selection and Feedback SNR Estimation}
\label{proposed}
In this section, we present the proposed CS-based feedback algorithm and outline the steps required to obtain the identity and the receive SNR of the strong users at the BS. We first demonstrate our approach in the simpler case of HD relaying and then extend it to the case of FD relaying.

\subsection{Feedback Threshold}
To apply CS theory, only a few strong users are allowed to feedback. This creates a sparse user regime. This is done by using a feedback threshold \footnote{Note that the threshold is the same for both relaying techniques HD and FD.}. The feedback threshold $\gamma_{\text{th}}$ is optimized to meet a target scheduling outage probability $\mathcal{P}_o$. The outage probability is defined as the probability that all users fail to report an SNR above $\gamma_{\text{th}}$. The scheduling outage probability can be calculated as
\begin{eqnarray}\label{poa}
\begin{split}
\mathcal{P}_o &= \mathbb{P}(\gamma_n < \gamma_{\text{th}}, \text{ for all } n=1,...,N)  \\
&= [F_{\gamma}(\gamma_{\text{th}})]^N 
\end{split}
\end{eqnarray}
where $S$ is the number of users that report an SNR above the threshold and  $\bar{S}=\mathbb{E}\left[S\right]$ is its mean value , $F_{\gamma}(\gamma_{\text{th}}) = 1-\exp({\frac{-\gamma_{\text{th}}}{\bar{\gamma}}})$ is the SNR CDF at the users side.  From (\ref{poa}), the feedback threshold can be calculated as
\begin{eqnarray} \label{zeta}
\gamma_{\text{th}} = F_{\gamma}^{-1}(\mathcal{P}_o^{1/N}),
\end{eqnarray}
where $\mathcal{P}_o$ is assumed to be very small. 

\subsection{User ID Estimation}
For the BS to make a user selection decision, it has to estimate the fed back vector $\mathbf{x}$. In what follows, we show that we can reliably estimate $\mathbf{x}$ and hence identify a set of strong users. More precisely, we show that in both feedback scenarios (HD and FD) recovering the fed back vector simplifies to solving an undertermined linear system.
\subsubsection{HD feedback}
In the case of HD feedback, the feedback is communicated over orthogonal channels through the two-hop network via the relay. In other words, the relay listens to the users' feedback in one time slot and forwards what it gets in an independent time slot. 
Thus from (\ref{hdreceive}), we have the following linear system : 
\begin{equation}\label{hdlinear}
\mathbf{y_s}=\mathbf{A}\mathbf{x}+\mathbf{\hat{z}},
\end{equation}
where $\mathbf{y_s}=\left[y_s^{(1)} \quad y_s^{(2)}  \quad ... \quad y_s^{(M)}\right]^t$, and $\mathbf{\hat{z}}=\left[\frac{w^{(1)}}{f}+z_r^{(1)} \quad \frac{w^{(2)}}{f}+z_r^{(2)}  \quad ... \quad \frac{w^{(M)}}{f}+z_r^{(M)}\right]^t$. Obviously, the noise vector is uncorrelated with covariance matrix $\mathbf{\Sigma}_{\mathbf{\hat{z}}}=\mathbb{E}\left[\mathbf{\hat{z}}\mathbf{\hat{z}}^t\right]$. To be able to determine $\mathbf{\Sigma}_{\mathbf{\hat{z}}}$, we need to evaluate $\mathbb{E}\left[\frac{1}{\left|f\right|^2}\right]$. We have \footnote{Note that for $K > 0$, $d=\frac{\left(K+1\right)^2}{2K+1}=1+\frac{K^2}{2K+1}$, so that $d > 1$.}
\begin{equation}
\begin{split}
\mathbb{E}\left[\frac{1}{\left|f\right|^2}\right]&=\int_0^{\infty}\frac{x^{d-2}}{\theta^d\Gamma(d)}e^{-\frac{x}{\theta}}dx \\
&=\frac{1}{\left(d-1\right)\theta}.
\end{split}
\end{equation}
Thus, $\mathbf{\Sigma}_{\mathbf{\hat{z}}}=\left(N_0+\frac{N_0}{\left(d-1\right)\theta}\right)\mathbf{I}_M$.
With this in mind, we can recover the support of $\mathbf{x}$ using regular compressive sensing recovery algorithms such as the least absolute shrinkage and selection operator (LASSO) outlined in \cite{cP08}.
\subsubsection{FD feedback}
In this case, we make use of the structure in the received signal in (\ref{ys}) and then applies the theory of Compressive Sensing to recover $\mathbf{x}$. Interestingly, the problem can be formulated as a block sparse recovery problem where the unknown vector is a block sparse vector with block size $J$. We denote by $\mathbf{x}_{(n)}=\left[x_n,\rho x_n,...,\rho^{J-1}x_n\right]^t$, $n=1,2,...,
N$, and by $\mathbf{a}^{j}_{c,n}=\left[ \underbrace{0,0,...,0}_{j},\mathbf{a}_{c,n}\left(1:M-j \right )^t \right ]^t$.

Then, based on (\ref{ys}), we have the following linear system

\begin{equation} \label{linear1}
\mathbf{y_s}=\mathbf{B}\mathbf{v}+\mathbf{\hat{z}} ,
\end{equation}
where 
$\mathbf{B}=\left[\mathbf{B}_{(1)},\mathbf{B}_{(2)},...,\mathbf{B}_{(N)}\right]$, $\mathbf{v}=\left[\mathbf{x}_{(1)}^t,\mathbf{x}_{(2)}^t,...,\mathbf{x}_{(N)}^t\right]^t$ and
$\mathbf{B}_{(n)}=\left[\mathbf{a}^{0}_{c,n},\mathbf{a}^{1}_{c,n},...,\mathbf{a}^{J-1}_{c,n}\right]$, $n=1,2,...,N$
 and $\mathbf{\hat{z}}=\left[\hat{z}_1 \quad \hat{z}_2  \quad ... \quad \hat{z}_M\right]^t$ the noise vector at the BS, where
\begin{equation}
\hat{z}_m = \begin{cases}
\frac{w^{(1)}}{f}+z_r^{(1)}, & m=1 \\
\frac{w^{(m)}}{f}+z_r^{(m)}+\rho z_r^{(m-1)}, & m=2,3,...,M. \\
\end{cases}
\end{equation}
It is not difficult to see that the noise at the BS is correlated with correlation matrix $\mathbf{\Sigma}_{\mathbf{\hat{z}}}=\mathbb{E}\left[\mathbf{\hat{z}}\mathbf{\hat{z}}^t\right]$. In the case where $J=3$, we have \footnote{Our analysis is applicable for any value of $J$, but for presentation purposes, we restrict the analysis to $J=3$.}

\begin{equation}
\label{tri}
\mathbf{\Sigma}_{\mathbf{\hat{z}}}=\begin{bmatrix}
\alpha_1 &\alpha_3 &0&\cdots&0 \\
 \alpha_3& \alpha_2 &\alpha_3 &\cdots& 0 \\
0 &\alpha_3 &\ddots &\ddots   &\vdots  \\
\vdots &\cdots &\ddots &\ddots &\alpha_3 \\
0&\cdots  &\cdots &\alpha_3&\alpha_2
\end{bmatrix},
\end{equation}
where $\alpha_1=N_0\left(1+\frac{1}{(d-1)\theta} \right )$, $\alpha_2=N_0\left(1+\rho^2+\frac{1}{(d-1)\theta}\right )$ and $\alpha_3=\rho N_0$. \\

As shown in (\ref{linear1}), the unknown vector $\mathbf{v}$ is block sparse with block size $J$ and sparsity $S$. Therefore, it is possible to apply the results in \cite{baraniuk} to reliably recover $\mathbf{v}$. As shown in \cite{baraniuk}, for reliable CS recovery, the number of measurements should scales as  $M=C\left(JS+S\log JN/S\right)$. Basically, the BS performs CS recovery (e.g. the CoSaMP algorithm in \cite{baraniuk}) on the resultant measurements vector $\mathbf{y}_s$ and estimates the ID of the strongest user. If at least one user is detected, the BS performs a subsequent SNR estimation and refinements to be discussed below, otherwise an outage is declared.
In the following subsection, we summarize an important result from CS theory for recovering the sparsity pattern of a block-sparse vector in a noisy setting.
\subsection{Block-Sparse Signal Recovery}\label{cs1}
The theory of compressed sensing permits efficient acquisition and reconstruction of a sparse signal (through multiplication by an appropriate random matrix) from only few measurements. For a sparse vector with size $N$ and sparsity $S$, it has been shown that robust signal recovery is possible from $M=\mathcal{O}\left(S\log N/S\right)$ \cite{spr}. One would expect that additional structure imposed on the unknown sparse signal would substantially reduce $M$ without sacrificing the recovery performance. One imposed structure could be the block-sparsity structure, where the locations of the significant entries in the sparse signal cluster in blocks. In the literature, the concept of recovering a block-sparse vector with reduced number of CS measurements has been studied in \cite{8} and \cite{43}. However robustness guarantees are restricted to either conventional CS sparse signals or recovery with noiseless measurements and do not have exact bounds for the required number of CS measurements. Recently in \cite{baraniuk}, Baraniuk \emph{et al.} showed that robustness guarantees can be achieved with $M=C\left(JS+S\log JN/S\right)$ CS measurements, where $C$ is a positive constant and $J$ is the cluster or the block size, which is a substantial improvement over $M=C\left(JS\log N/S\right)$ that would be required by conventional CS recovery algorithms. For more analytical details about block-spare recovery, readers are referred to \cite{8}-\cite{baraniuk} and references therein.
\subsection{User SNR Estimation}
\subsubsection{HD feedback}
After CS recovery, the BS obtains information on the location of the non-zeros in $\mathbf{x}$ or equivalently the support of $\mathbf{x}$ denoted by $\mathcal{S}$, where $\left|\mathcal{S}\right|=S$ is the cardinality of $\mathcal{S}$. Therefore, the linear system in (\ref{hdlinear}) can be rewritten as 
\begin{equation}\label{hd}
\mathbf{y_s}=\mathbf{A}_{\mathcal{S}}\mathbf{x}_{\mathcal{S}}+\mathbf{\hat{z}}
\end{equation}
We can now apply least squares (LS) to estimate the entries of $\mathbf{x}_{\mathcal{S}}$ as follows 
\begin{equation}\label{ls}
\begin{split}
\hat{\mathbf{x}}_{\mathcal{S}}&=\left(\mathbf{A}_{\mathcal{S}}^t\mathbf{A}_{\mathcal{S}}\right)^{-1}\mathbf{A}_{\mathcal{S}}^t\mathbf{y_s} \\
&=\mathbf{x}_{\mathcal{S}}+\mathbf{e}_
{LS},
\end{split}
\end{equation}
where $\mathbf{e}_{LS}$ is the estimation error vector. Upon conditioning on $\mathbf{A}_{\mathcal{S}}$, vector $\mathbf{e}_{LS}$ is Gaussian since it results from a linear operation on Gaussian random variables. 
The covariance of  $\mathbf{e}_{LS}$ is given by $\boldsymbol{}\left[\left(\mathbf{A}_{\mathcal{S}}^t\mathbf{A}_{\mathcal{S}}\right)^{-1}\right]$. As $M$ grow to infinity while $S$ is fixed, $\left(\mathbf{A}_{\mathcal{S}}^t\mathbf{A}_{\mathcal{S}}\right)^{-1}$ becomes close to its mean $\frac{M}{M-S-1}N_0{\bf I}_S$ \cite{icassp} in the sense that :
\begin{equation}\label{sighd}
\left(\mathbf{A}_{\mathcal{S}}^t\mathbf{A}_{\mathcal{S}}\right)^{-1}-\frac{M}{M-S-1}N_0{\bf I}_S\xrightarrow[M\to+\infty]{a.s} 0
\end{equation}
The convergence holds in the almost sure sense and shows that the variance of the error is almost the same for all vector entries. In light of this result, we will consider for simplicity that ${\bf e}_{LS}$ is Gaussian vector whose entries are indendenent Gaussian variables with zero-mean and variance $\sigma_{e,LS}^2=\frac{M}{M-S-1}N_0$. 
\subsubsection{FD feedback}
Similarly, in the case of FD feedback, the BS recovers the support of $\mathbf{v}$ denoted by $\mathcal{I}$, where $\left|\mathcal{I}\right|=JS$. Therefore, the linear system in (\ref{linear1}) can be rewritten as
\begin{equation}\label{eq}
\mathbf{y}_s=\mathbf{B}_{\mathcal{I}}\mathbf{v}_{\mathcal{I}}+\mathbf{\hat{z}},
\end{equation}
 Since the covariance matrix $\mathbf{\Sigma}_{\mathbf{\hat{z}}}$ is correlated, one can apply the Best Linear Unbiased Estimator (BLUE) \footnote{When the noise is white as in eq. (\ref{hd}), LS and BLUE are equivalent.} to estimate the entries of $\mathbf{v}_{\mathcal{I}}$ as follows \cite{Sayed}
\begin{equation} \label{blue}
\begin{split}
\hat{\mathbf{v}}_{\mathcal{I}}&=\left(\mathbf{B}_{\mathcal{I}}^t\mathbf{\Sigma}_{\mathbf{\hat{z}}}^{-1}\mathbf{B}_{\mathcal{I}}\right)^{-1}\mathbf{B}_{\mathcal{I}}^t\mathbf{\Sigma}_{\mathbf{\hat{z}}}^{-1} \mathbf{y}_s \\
&=\mathbf{v}_{\mathcal{I}}+\mathbf{e}_{BLUE},
\end{split}
\end{equation}
where $\mathbf{e}_{BLUE}$ is the BLUE estimation error. Similarly to the HD feedback scenario, upon conditioning on  ${\bf B}_{\mathcal{I}}$, vector $\mathbf{e}_{BLUE}$ is Gaussian with zero-mean and covariance $\left({\bf B}_{\mathcal{I}}^{t}\mathbf{\Sigma}_{\mathbf{\hat{z}}}^{-1}{\bf B}_{\mathcal{I}}\right)^{-1}$.
When $M$ tends to infinity with $S$ and $J$ fixed, we can prove that:
\begin{lemma}
Assume that $M\to+\infty$ with $S$ and $J$ fixed. Assume that $\lim\sup_{M}\left\|\mathbf{\Sigma}_{\mathbf{\hat{z}}}^{-1}\right\|<+\infty$. Then,
\begin{equation}
\left({\bf B}_{\mathcal{I}}^{t}\mathbf{\Sigma}_{\mathbf{\hat{z}}}^{-1}{\bf B}_{\mathcal{I}}\right)^{-1}-\frac{M{\bf I}_{JS}}{tr\left(\mathbf{\Sigma}_{\mathbf{\hat{z}}}^{-1}\right)} \xrightarrow[M\to+\infty]{a.s.}0.
\end{equation}
\end{lemma}
\begin{proof}
See Appendix A for proof.
\end{proof}\\
Again, we note that the above convergence suggests considering  the noise at the output of the BLUE estimator as a Gaussian random vector whose entries are independent with zero-mean and variance $\sigma_{e,BLUE}^2$ given by:
\begin{equation}\label{sigfd}
\sigma_{e,BLUE}^2 = \frac{M}{tr\left(\mathbf{\Sigma}_{\mathbf{\hat{z}}}^{-1}\right) } 
\end{equation}

\subsection{SNR Back-Off}
As stated in (\ref{ls}) and (\ref{blue}), the LS and BLUE estimators are noisy which means that the estimated SNR can be higher or lower than the actual one. This is problematic, since an estimated SNR higher than the actual one results in a transmission rate higher than the maximum rate the end user can support. To deal with this, we back-off on the estimated noisy SNRs. From (\ref{ls}) and (\ref{blue}), each entry of $\hat{\mathbf{x}}_{\mathcal{S}}$ and $\hat{\mathbf{v}}_{\mathcal{I}}$ respectively can be represented in a scalar equation as $\hat{\gamma}=\gamma+e$, where $\gamma$ and $\hat{\gamma}$ stands for the actual and the estimated SNRs respectively, and $e$ is the Gaussian error. Our back-off strategy simply subtracts an amount $\Delta$ from
$\gamma$. Hence, $\hat{\gamma}$ becomes
\begin{equation}
\hat{\gamma}=\gamma+e-\Delta
\end{equation}
To characterize the outage when the estimated SNR is higher than the actual one, we define $\eta$ as the back-off efficiency which is simply given by
\begin{equation}
\label{eta}
\begin{split}
\eta=&\mathbb{P}\left(\hat{\gamma} \leq \gamma\right) \\
=& \mathbb{P}\left(e \leq \Delta\right) \\
=& 1-Q\left(\frac{\Delta}{\sigma_e}\right),
\end{split}
\end{equation}
where $\sigma_e$ is either the LS or the BLUE error variance. Obviously, the performance of the proposed algorithm will depend on $\Delta$. A low value of $\Delta$ may result in an outage while a high value will result in a low rate. The optimal value will be derived in the following section.

\section{Performance Analysis} \label{PA}

The performance of the proposed feedback algorithm will be evaluated based on three performance criteria: i) the feedback load, ii) the achievable rate, and iii) the achievable throughput. Obviously, each of these performance criteria will depend on the feedback relaying mode, i.e. HD or FD feedback. HD feedback results in lower noise variance at the BS at the expense of a larger feedback overhead, while FD feedback results in a higher noise variance at the BS (due to the relay self-interference) but results in a lower feedback overhead.


\subsection{Feedback Load}
The feedback load $L$ is defined as the total number of feedback mini-slots required for the BS to make a user scheduling decision. In other words, $L$ is the total number of measurements required to have robust CS recovery at the BS and it is given by
\[
L_{\text{HD}} =2C\bar{S}\log N/\bar{S}
\]
in the case of HD feedback\footnote{In HD feedback ($J=1$), the relay receives and forwards the users' feedback information using two orthogonal channels. Therefore, the BS requires twice the number of mini-slots for feedback reception when compared to FD feedback.}, and  
\[
L_{\text{FD}}=C\left(J\bar{S}+\bar{S}\log JN/\bar{S}\right),
\] 
in the case FD feedback \cite{baraniuk}. The term $\bar{S}$ represents the average number of users that feedback and it is given by
\begin{equation}
\begin{split}
\bar{S}&=\sum_{n=1}^Nn\binom{N}{n}\left(1-F_{\gamma}\left(\gamma_{th}\right)\right)^nF_{\gamma}\left(\gamma_{th}\right)^{N-n} \\
&=N\left(1-\mathcal{P}_0^{1/N}\right).
\end{split}
\end{equation}

\subsection{Achievable Rate}

The achievable rate for the HD and FD cases can be expressed as
\begin{equation} \label{capacity}
\begin{split}
\mathcal{R} = \mathbb{E} \bigg[ \log\left(1+\gamma_{eq}-\Delta\right)\left(1-\mathcal{P}_o\right)\left(1-Q\left(\frac{\Delta}{\sigma_e}\right)\right)\bigg],
\end{split}
\end{equation}
where 
\begin{equation}
\gamma_{eq}=\min \left(\frac{P_s\left|f\right|^2}{P_r\rho^2+N_0},\gamma_{n^*}\right),
\end{equation}
where $n^*$ is defined in eq. (\ref{SNR}), and $\sigma_e^2=\sigma_{e,LS}^2$ in the case of HD feedback, $\sigma_e^2=\sigma_{e,BLUE}^2$ otherwise. Using the Jensen's inequality, (\ref{capacity}) can be upper bounded as \footnote{Note that the inequality comes also from the fact that we didn't include the CS recovery probability in the rate expression. This probability is not available in closed forms for block-CS recovery.}
\begin{equation} \label{capacity1}
\begin{split}
\mathcal{R} \leq  \log\left(1+\mathbb{E}[\gamma_{eq}]-\Delta\right)\left(1-\mathcal{P}_o\right)\left(1-Q\left(\frac{\Delta}{\sigma_e}\right)\right).
\end{split}
\end{equation}

It is clear from (\ref{capacity1}) that the achievable rate depends on the back-off value $\Delta$. To obtain the optimal value of $\Delta$, we maximize the achievable rate with respect to $\Delta$. To do that, we differentiate the upper bound on $\mathcal{R}$ with respect to $\Delta$ and equate the result to zero. We denote by $\Delta^*$, the optimal back-off, then $\Delta^*$ satisfies
\begin{equation}\label{del}
\begin{split}
&\left(\frac{1+\bar{\gamma}_{eq}-\Delta^*}{\sqrt{2\pi}\sigma_e}\right)\exp\left(-\frac{(\Delta^*)^2}{2\sigma_e^2}\right)\log\left(1+\bar{\gamma}_{eq}-\Delta^*\right) \\
&+Q\left(\frac{\Delta^*}{\sigma_e}\right)=1
\end{split}
\end{equation}
where, $\bar{\gamma}_{eq}=\mathbb{E}\left[\gamma_{eq}\right]$.  Then, the value of $\Delta^*$ that satisfies (\ref{del}) is used in (\ref{capacity}).

It now remains to derive the average end-to-end SNR at the strongest user, i.e. $\mathbb{E}[\gamma_{eq}]$. To achieve this, we first need to derive the PDF of $\gamma_{eq}$, $p_{\gamma_{eq}}\left(x\right)$. Since $\gamma_{eq}$ is simply the minimum of two independent random variables, its CDF can be derived as follows
\begin{equation}
\begin{split}
F_{\gamma_{eq}}\left(x\right)= &1-\left[1-F_{\left|f\right|^2/\left(\rho^2+N_0\right)}\left(x\right)\right]\left[1-F_{\gamma^*}\left(x\right)\right] \\
=& 1-\left[1-\frac{\gamma\left(d,\frac{\lambda x}{\theta}\right)}{\Gamma\left(d\right)}\right]\left[1-\left(1-e^{-\frac{x}{\bar{\gamma}}}\right)^N\right]
\end{split}
\end{equation}
where $\lambda=\rho^2+N_0$.
Therefore, $p_{\gamma_{eq}}\left(x\right)$ is given by
\begin{equation}
\begin{split}
p_{\gamma_{eq}}\left(x\right)=& \frac{\partial F_{\gamma_{eq}}\left(x\right)}{\partial x} \\
=& \underbrace{\frac{N}{\bar{\gamma}\Gamma(d)}e^{-\frac{x}{\bar{\gamma}}}\left(1-e^{-\frac{x}{\bar{\gamma}}}\right)^{N-1}\Gamma\left(d,\frac{\lambda x}{\theta} \right )}_{G_1\left(x\right)} \\
&+\underbrace{\left[1-\left(1-e^{-\frac{x}{\bar{\gamma}}}\right)^N\right]\frac{1}{\Gamma(d)} \left(  \frac{\lambda}{\theta}\right)^d x^{d-1}e^{-\frac{\lambda x}{\theta}}}_{G_2\left(x\right)}
\end{split}
\end{equation}
Hence,
\begin{equation}\label{pdf}
\begin{split}
\bar{\gamma}_{eq}= &\int_0^\infty xp_{\gamma_{eq}}\left(x\right) dx \\
=& \underbrace{\int_0^\infty xG_1\left(x\right) dx}_{\mu_1}+ \underbrace{\int_0^\infty xG_2\left(x\right) dx}_{\mu_2}
\end{split}
\end{equation}
To be able to derive the integrals in (\ref{pdf}), we use the fact that $\left(1-e^{-\frac{x}{\bar{\gamma}}}\right)^N=\sum _{n=0}^N \binom{N}{n}\left( -1\right)^n e^{-\frac{n x}{\bar{\gamma}}}$ and that $\int _0^\infty x^{\mu-1}e^{-\beta x}\Gamma\left( \nu,\alpha x\right )dx=\frac{\alpha^\nu \Gamma(\mu+\nu)}{\mu\left(\alpha+\beta \right )^{\mu+\nu}}\quad _2F_1\left(1,\mu+\nu;\mu+1;\frac{\beta}{\alpha+\beta} \right ) $ \cite{integral}, where $ _2F_1\left(.,.;.;. \right ) $ is the Gaussian hypergeometric function. Therefore, the first integral in (\ref{pdf}) can be evaluated as
\begin{equation}
\begin{split}
\mu_1&=\frac{N\lambda^d}{2\bar{\gamma}d(d+1)\theta^d}\sum_{n=0}^{N-1}\binom{N-1}{n}\left(-1 \right )^n  \left( \frac{\lambda}{\theta}+\frac{n+1}{\bar{\gamma}}\right )^{-d-2}
\\
&\times _2F_1\left(1,d+2;3;\frac{\left(n+1 \right )/\bar{\gamma}}{\lambda/\theta+\left(n+1 \right )/\bar{\gamma}} \right ).
\end{split}
\end{equation}
Similarly, the second integral in (\ref{pdf}) can be calculated as
\begin{equation}
\mu_2= d\left(\frac{\lambda}{\theta} \right )^d\sum _{n=1}^N\binom{N}{n}\left( -1\right )^{n+1}\left(\frac{\lambda}{\theta}+\frac{n}{\bar{\gamma}} \right )^{-d-1}
\end{equation}
Finally,
\begin{equation}\label{receiveSNR}
\bar{\gamma}_{eq}=\mu_1+\mu_2
\end{equation}

\subsection{Achievable Throughput}
In the previous section, it was assumed that the amount of air-time reserved to feedback is negligible as compared to the transmission time. This assumption does not give much insight for practical scenarios. In this paper, we define the achievable throughput as the number of transmitted bits per unit time (bps/Hz). The network throughput can be explicitly given by
\begin{equation}
\begin{split}
\mathcal{T}&=\mathcal{R} \frac{\left(T_c-LT_{ms}\right)}{T_c} \\
&=\mathcal{R}\left(1-L\tau\right)
\end{split}
\end{equation}
where $T_c$ is the channel coherence time, $T_{ms}$ is the time needed to transmit one feedback mini-slot and $\tau=\frac{T_{ms}}{T_c}$ is the normalized mini-slot time (see Figure 2).

\section{Numerical Results}\label{nr}

For all simulations, we set $C=2$ and we use the CoSaMP algorithm proposed in \cite{baraniuk}\footnote{CoSaMP  recovery algorithm that we use to recover $\mathbf{v}$ is available at http://dsp.rice.edu. For HD feedback, the group size becomes 1.} for CS recovery. Unless otherwise specified, we set: $P_s=P_r=1$, $b^2=20$ dB, $\sigma^2=0$ dB, $\sigma^2_g=5$ dB, $N_0=-15$ dB, $\mathcal{P}_0=0.01$, and $J=3$.
 Throughout this section, we compare the proposed scheme (scheme. 1) with the following schemes:
\begin{itemize}
\item Noiseless Dedicated Feedback (scheme. 2): All users feedback in independent noise-free, interference-free feedback channels (via the relay).
\item Random User Selection (scheme. 3): For each transmission period, a user is randomly selected and allocated the downlink channel resource.
\item Proposed without back-off (scheme. 4): The same as scheme. 1), but without any back-off (SNR refinement) strategy.
\end{itemize}

\begin{figure}[t!]
        \centering
    \includegraphics[width=3.5in]{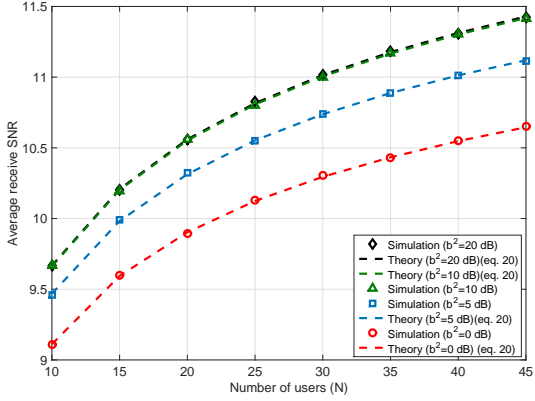}
\caption{Average Receive SNR versus the number of users ($N$), $\sigma^2=0$ dB, $\sigma^2_g=5$ dB, $N_0=-15$ dB}
\label{mmm}
\end{figure}
\begin{figure}[t!]
        \centering
   \includegraphics[width=3.5in]{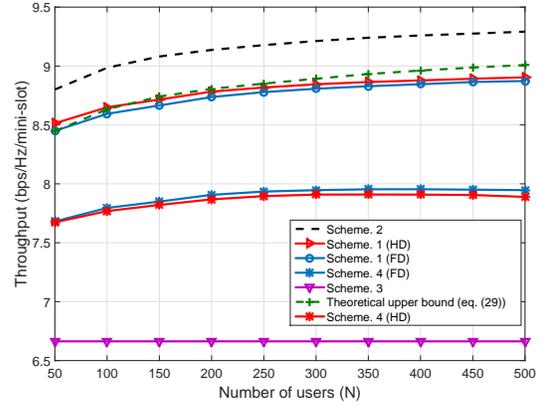}
\caption{Average Rate versus the number of users (N), $b^2=20$ dB, $\sigma^2=0$ dB, $\sigma^2_g=5$ dB, and $\mathcal{P}_0=10^{-2}$.}
\label{fig3}
\end{figure}

\begin{figure}[t!]
        \centering
    \includegraphics[width=3.5in]{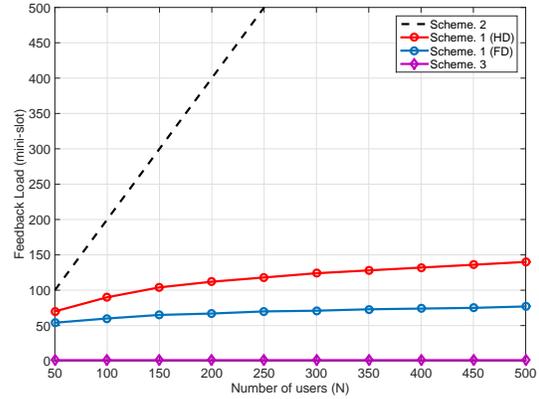}
\caption{Average Feedback Load versus the number of users (N)}
\label{fig4}
\end{figure}

\begin{figure}[t!]
        \centering
    \includegraphics[width=3.5in]{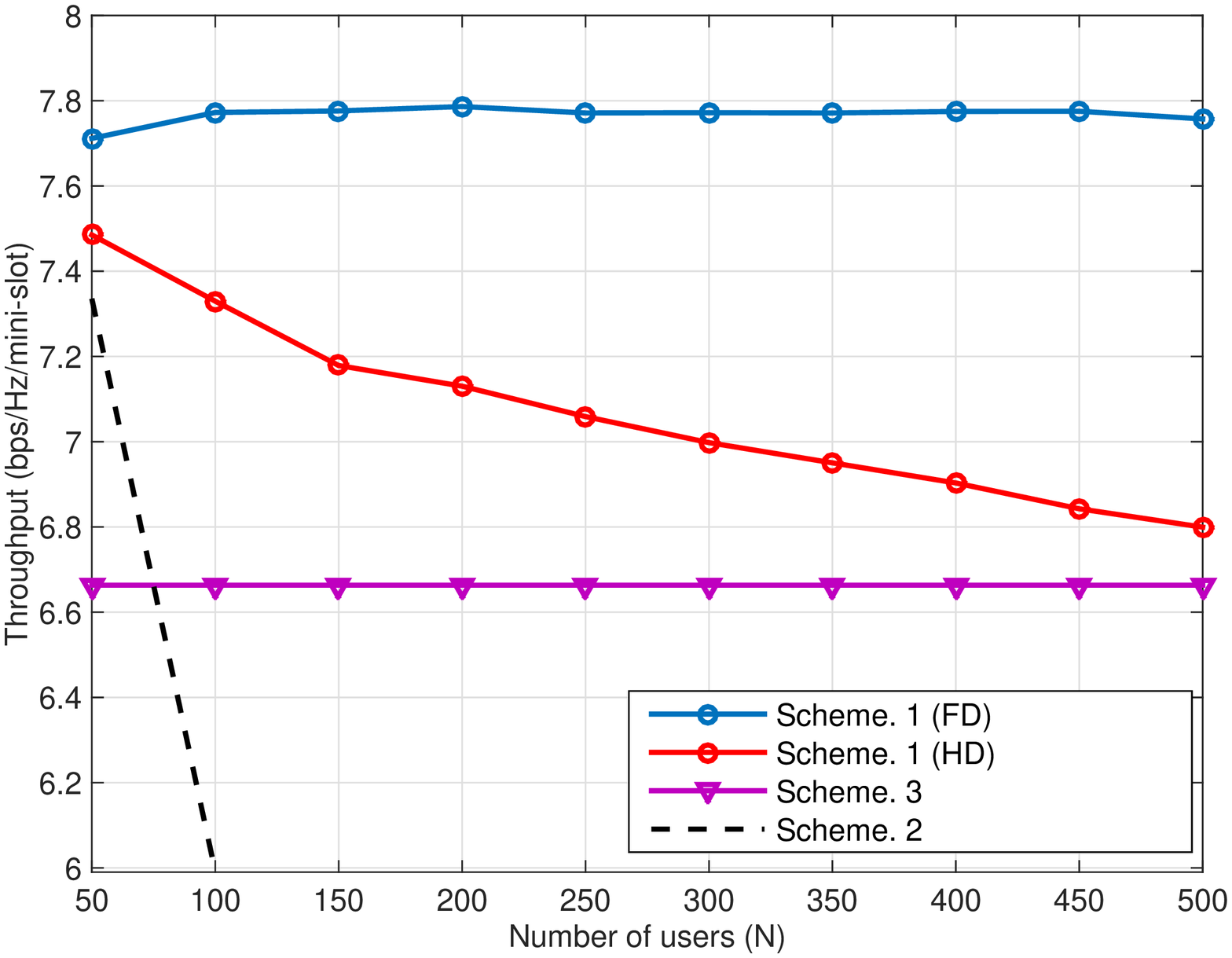}
\label{fig5}
\caption{Average Throughput versus the number of users (N), $\tau=1/600$}
\end{figure}

\begin{figure}[t!]
        \centering
    \includegraphics[width=3.5in]{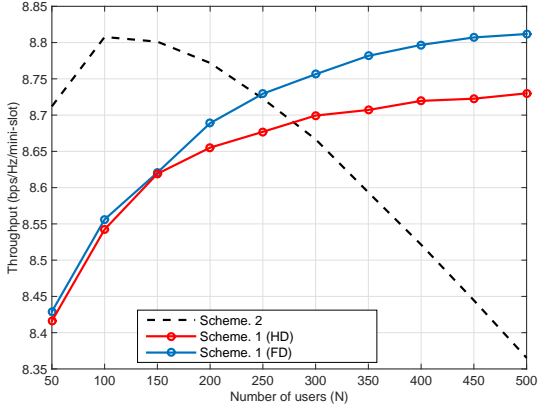}
\caption{Average Throughput versus the number of users (N), $\tau=10^{-4}$}
\label{fig6}
\end{figure}

In Figure \ref{mmm}, we compare the theoretical result obtained in eq. (\ref{receiveSNR}) with the numerical (simulated) average SNR. For all values of the average LOS component $b^2$, the theoretical and the numerical results exactly match, which validates the formula in eq. (\ref{receiveSNR}). It is also clear from Figure \ref{mmm} that increasing $b^2$ increases the average receive SNR as well. However, we notice that the average receive SNR converges for $b^2 \geq 10$ dB, this can be explained by the fact that the receive SNR is constrained by the SNR of the second hop.
In Figure \ref{fig3}, it is clear in one hand that scheme. 2 offers the best performance in terms of average rate, which is expected since perfect channel state information (CSI) is available at the BS. In the other hand, both versions of the scheme. 1 (HD and FD feedback) offer a comparable performance with a gap of 0.42 bps/Hz as compared to scheme. 2, where the HD version is slightly better than the FD version as it is free from interference.  Figure \ref{fig3} also shows that the performance of the proposed scheme without back-of (scheme. 4) suffers a rate-hit as compared to scheme. 1 and scheme. 2, and hence, the BS has to perform a back-off strategy in order to account for the feedback noise. Random user selection (scheme. 3) is shown to yield the worst performance since no feedback is communicated to the BS.

In Figure \ref{fig4}, we evaluate the performance in terms of feedback load, where we assume that the amount of time needed for channel estimation is negligible compared to the feedback air-time. As shown in Figure \ref{fig4}, scheme. 1 consumes much less feedback as compared to scheme. 2 especially at high number of users $N$. This is expected, since scheme. 2 requires the feedback to grow linearly with $N$, whereas scheme. 1 requires the feedback to grow logarithmically with $N$. The feedback load of scheme. 1 with HD feedback is worse than that  of scheme. 1 with FD feedback since feedback is performed over the two hops. 

To have more insight, we simulate the network throughput for $\tau=1/600$ in Figure 6. As shown, scheme. 1 offers the best throughput, and scheme. 2 yields the worst throughput since it requires a large amount of feedback. The throughput of the HD version of scheme. 1 is shown to deteriorate with the number of users. The reason for this is that HD feedback requires more feedback time as the number of users grow when compared to FD feedback. In Figure \ref{fig6}, we decrease $\tau$ to be $\tau=10^{-4}$ (i.e. higher coherence time). We first notice that at low number of users ($N < 250$), scheme. 2 performs better than scheme. 1 (due to the reduced value of $\tau$) and the HD version of scheme. 1 performs slightly lower than scheme. 1 (FD). At high number of users, scheme. 1 (FD) performs the best. Again this is due to its lower feedback requirements.

\section{Conclusion}\label{conclusion}
In this paper, we proposed an opportunistic CS-based feedback algorithm for a relay-aided multiuser network. Instead of allocating a feedback channel for each user, all users are allocated a pool of shared channels for feedback transmissions. We derived the end-to-end receive SNR at the destination user and the post-CS detection/refinement error variance. Moreover, we considered both HD feedback and FD feedback and showed that FD feedback offers comparable rate performance when compared to HD feedback. In addition to this, FD feedback results in lower feedback load, especially at high number of users, and therefore, offers better throughput than HD feedback.

\section*{Appendix A}
\section*{Proof of lemma 1}
 Let $\mathbf{H}=\mathbf{B}_\mathcal{I}^t\mathbf{\Sigma}_{\mathbf{\hat{z}}}^{-1}\mathbf{B}_\mathcal{I}$. Matrix ${\bf H}$ is of finite size $JS$. It thus suffices to study the convergence of the elements of ${\mathbf{H}}$.
Using results on the convergence of quadratic forms \cite[Lemma5,Lemma4]{wagner}, we can establish  that the off-diagonal entries of ${\bf H}$ given by $$\left(\mathbf{a}_{c,n}^{j}\right)^t\mathbf{\Sigma}_{\mathbf{\hat{z}}}^{-1}\mathbf{a}_{c,k}^{l},\forall (n,j)\neq (k,l)$$ converge almost surely to zero, while the diagonal entries converge almost surely to their mean, i.e, 
$$
\left(\mathbf{a}_{c,n}^{j}\right)^t\mathbf{\Sigma}_{\mathbf{\hat{z}}}^{-1}\mathbf{a}_{c,n}^{j}-d_{j,n}\xrightarrow[M\to+\infty]{a.s.}0.
$$
with $d_{j,n}\triangleq\mathbb{E} \left(\mathbf{a}_{c,n}^{j}\right)^t\mathbf{\Sigma}_{\mathbf{\hat{z}}}^{-1}\mathbf{a}_{c,n}^{j}$ reading as:
$$
d_{j,n}=\frac{1}{M}\left(tr\left(\mathbf{\Sigma}_{\mathbf{\hat{z}}}^{-1}\right)-\sum_{k=1}^j \left[\mathbf{\Sigma}_{\mathbf{\hat{z}}}^{-1}\right]_{k,k}\right)
$$
In the expression of $d_{j,n}$, the dominant term is $\frac{1}{M}tr \left(\mathbf{\Sigma}_{\mathbf{\hat{z}}}^{-1}\right)$ while the second term converges almost surely to zero, since
$$
\frac{1}{M}\sum_{k=1}^j \left[\mathbf{\Sigma}_{\mathbf{\hat{z}}}^{-1}\right]_{k,k} \leq \frac{J\|\mathbf{\Sigma}_{\mathbf{\hat{z}}}^{-1}\|}{M}\leq \frac{J}{M}\lim\sup \|\mathbf{\Sigma}_{\mathbf{\hat{z}}}^{-1}\|
$$
We therefore obtain:
$$
\mathbf{H}^{-1}-\frac{M {\bf I}_{JS}}{tr\left(\mathbf{\Sigma}_{\mathbf{\hat{z}}}^{-1} \right)} \xrightarrow[M\to+\infty]{a.s.}0,
$$

In Figure \ref{fig7}, we plot the exact and almost sure limit (lemma 1) of the noise variance versus $M/\bar{S}$. From the figure we notice that the lower bound on $\sigma_{e,BLUE}^2$ is tight for moderate to high ratio of $M/\bar{S}$.  This concludes the proof of the lemma.

\begin{figure}[t!]
        \centering
    \includegraphics[width=3.5in]{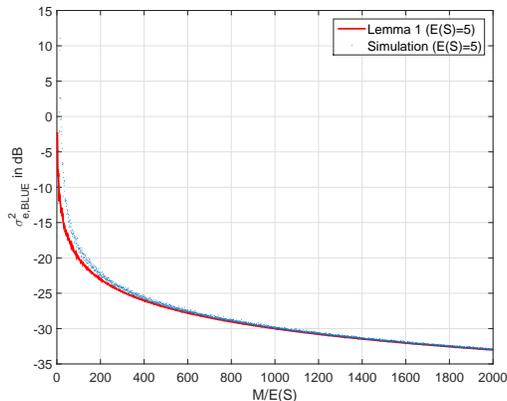}
\caption{BLUE error variance: Exact versus theoretical limit (Lemma 1).}
\label{fig7}
\end{figure}

{}
%

\end{document}